%% file: Seit2014.tex
\documentclass{sig-alternate-10pt}

\usepackage{array}
\usepackage{graphics}
\usepackage{epsfig}
\usepackage{color}
\usepackage{url} 
\usepackage{caption}
\usepackage{subcaption}
\usepackage[bookmarks,colorlinks=true,citecolor=dkgreen]{hyperref}
\usepackage{multirow}

\newtheorem{prop}{Proposition}

\usepackage{listings}

\usepackage{balance}
\usepackage[sort]{cite}
\usepackage{multirow}
\usepackage[table]{xcolor}

\definecolor{dkgreen}{rgb}{0,0.42,0}
\definecolor{gray}{rgb}{0.5,0.5,0.5}
\definecolor{mauve}{rgb}{0.58,0,0.82}

\title{The Cloud Needs a Reputation System}
\author
{
Murad Kablan$^\dagger$
\and Carlee Joe-Won$^\mp$
\and Sangtae Ha$^\dagger$
\and Hani Jamjoom$^\ddagger$ 
\and Eric Keller$^\dagger$
\end{tabular}\newline\begin{tabular}{c c c c c}
        \affaddr{$^\dagger$University of Colorado} & ~~~~~ & \affaddr{$^\ddagger$IBM Watson Research Center} & ~~~~~ & \affaddr{$^\mp$Princeton University}\\
       \affaddr{Boulder, CO, USA} & ~~~~~ & \affaddr{Yorktown Heights, NY, USA} & ~~~~~ & \affaddr{Princeton, NJ, USA}\\
}

\newcommand{\nip}[1]{\vspace{1ex}\noindent\textbf{#1}}

\begin{document}
\sloppy

\maketitle
\thispagestyle{empty}

\renewcommand\_{\textunderscore\allowbreak}
\newcommand{\paragraphb}[1]{\vspace{0.03in}\noindent{\bf #1} }
\newcommand{\paragraphe}[1]{\vspace{0.03in}\noindent{\em #1} }
\newcommand{\paragraphbe}[1]{\vspace{0.03in}\noindent{\bf \em #1} }
\newcommand{\comment}{\textcolor{red}}
\newcommand{\eric}{\textcolor{red}}
\newcommand{\murad}{\textcolor{blue}}

\newcommand{\eg}{{\em e.g., }}
\newcommand{\ie}{{\em i.e., }}

\input{abstract}
\input{intro}

\input{related}

\input{motivation}
\input{arch}
\input{policiesandinterfaces}

\input{ibr}

\input{impl}

\input{properties}
\input{eval}
\input{concl}

\footnotesize
\bibliographystyle{abbrv}
\balance
\bibliography{refs,ibm}

\end{document}

%% file: abstract.tex
\begin{abstract}

Today's cloud apps are built from many diverse services that are
managed by different parties.  At the same time, these parties, which
consume and/or provide services, continue to rely on arcane static
security and entitlements models.  In this paper, we introduce {\em
  Seit}, an inter-tenant framework that manages the interactions
between cloud services. Seit is a software-defined reputation-based
framework. It consists of two primary components: (1) a set of
integration and query interfaces that can be easily integrated into
cloud and service providers' management stacks, and (2) a controller
that maintains reputation information using a mechanism that is
adaptive to the highly dynamic environment of the cloud.  We have
fully implemented Seit, and integrated it into an SDN controller, a load
balancer, a cloud service broker, an intrusion detection system, and a
monitoring framework.  We evaluate the efficacy of Seit using both an
analytical model and a Mininet-based emulated environment. Our
analytical model validate the isolation and stability properties of
Seit.  Using our emulated environment, we show that Seit can provide
improved security by isolating malicious tenants, reduced costs by
adapting the infrastructure without compromising security, and
increased revenues for high quality service providers by enabling
reputation to impact discovery.

\end{abstract}

%% file: intro.tex
\section{Introduction}


Building and deploying any distributed ``app'' today is radically different from a
decade ago. Where traditional applications of the past required
dedicated infrastructure and middleware stacks, today's apps not only
run on shared---cloud---infrastructure, they rely on many services,
residing within and outside of underlying cloud. An app, for example,
can use Facebook for authentication, Box for storage, Twilio for
messaging, Square for payments, Google AdSense for advertising, etc. This
trend of deploying and consuming services (often referred to as the
{\em mesh economy}) can be seen by the rapid growth of cloud platforms which integrate services (and not just compute) like
Amazon Web Services~\cite{aws}, Microsoft Azure~\cite{windowsazure}, Heroku~\cite{heroku}, IBM
BlueMix~\cite{bluemix}, and CloudFoundry~\cite{cloudfoundry}, to name a
few. More importantly, these emerging platforms further encourage the
construction of apps from even smaller---micro---services (\eg GrapheneDB,
Redis, 3scale, etc.), where such services are developed and managed by
different tenants.


Despite the shift in how applications are being disaggregated, 
the management of security between services and tenants remains
largely the same: {\em one focused on a perimeter defense with a largely
static security configuration}.  This is exacerbated by the isolation
mechanisms being provided by cloud providers. It is also reinforced 
by the research community, which has also focused on
technologies that ensure isolation \cite{Shieh2011seawall, Guo2010secondnet, Popa2012faircloud, Varadarajan2014scheduler}.
While isolation enables tenants to reason about
their perimeter, perimeter defense is widely considered
insufficient~\cite{Zhang2012crossvm, Ristenpart2009heyyou, Zhang2014paasside, Wool2004, Bartal2004, Kaplan2014, thomson2014darkreading}.
Attackers are largely
indistinguishable from innocent parties and can
rely on their relative anonymity to bypass the perimeter.
Furthermore, erecting virtual perimeters in the cloud wastes an
opportunity for efficiency optimizations available in a multi-tenant
cloud infrastructure, especially since inter-tenant communication has
been shown to be an important component in intra-cloud data center
traffic~\cite{Ballani2013}. Such intra-cloud traffic will, of course, only grow as
apps move onto Platform as a Service (PaaS) clouds like CloudFoundry,
Heroku, and IBM BlueMix.


The insufficient nature of static security configuration has received
some attention from the research community in the form of highly
programmable network infrastructures such as software-defined
networking (SDN)~\cite{casado2007ethane, openflow} and network function virtualization
(NFV) or software-defined middlebox infrastructure~\cite{NFV,Tennenhouse97asurvey, simple,
comb, xomb, mbox_models, Gember-Jacobson2014opennf}.  To date, existing research has largely focused
on the systems to enable a dynamic security infrastructure, but leave
the automated use of these newly programmable infrastructures as an
open topic.




In this paper, we argue that the cloud needs a reputation system.
Reputation systems leverage the existence of many, collaborating
parties to automatically form opinions about each other.  A cloud can
leverage the power of the crowd---the many interacting tenants---to
achieve three primary goals: (1) focus isolation-related resources on
tenants that more likely cause problems, (2) optimize communication
between behaving parties, and (3) enable a security posture that
automatically adapts to the dynamicity of tenants and services
entering and leaving a cloud infrastructure.



In addition to the above three goals, we believe a reputation-based
system can encourage a culture of self-policing.  A report from
Verizon~\cite{verizon} shows that a large majority of compromises are
detected by a third party, not the infected party itself, as the
infected software starts interacting with external services.  In
service-centric clouds, being able to monitor sentiment through the
reputation system will allow a good intentioned tenant to know
something is wrong (e.g., when its reputation drops).  This is a
missing feature in traditional, security-centric infrastructures that
are based on isolation from others.

We introduce Seit\footnote{Seit means reputation or renown in the
  Arabic language.}, a general reputation-based framework. Seit
defines simple, yet generic interfaces that can be easily integrated
into different cloud management stacks. These interfaces interact with
a reputation manager that maintains tenant reputations and governs
introductions between different tenants. Together, the interfaces and
reputation manager enable reputation-based service differentiation in a way that
maintains stable reputations at each tenant and separates misbehaving
tenants from well-behaved tenants. 
Specifically, this paper makes the following contributions: 

\begin{itemize}

\item An \emph{architecture} that represents a practical realization
  of an existing reputation mechanism that is resilient to common
  attacks on reputation systems and adapted to support the operating
  environment of the cloud's dynamically changing behavior.  We
  optimize that mechanism with a query mechanism that supports the
  abstraction of a continuous query being performed.

\vspace{-1ex}
\item The demonstration of the feasibility with a \emph{prototype}
  implementation and \emph{integration} of Seit across a number of
  popular cloud and network components, including:
  Floodlight~\cite{floodlight} (an SDN controller),
  CloudFoundry~\cite{cloudfoundry} (a Platform as a Service cloud),
  HAProxy~\cite{haproxy} (a load balancer), Snort~\cite{snort} (an
  intrusion detection system), and Nagios~\cite{nagios} (a monitoring
  system).

\vspace{-1ex}
\item A proof, through an \emph{analytical model}, that the system is
  able to effectively isolate bad tenants, and that the system can
  remain stable despite the high dynamics.

\vspace{-1ex}
\item The demonstration of the effectiveness of Seit and the cloud
  components we integrated through an \emph{evaluation} that shows (1)
  an effective ability to protect tenants from attackers by
  propagating information about malicious behavior by way of
  reputation updates, (2) an effectiveness in reducing costs of
  security middlebox infrastructure without compromising security, and
  (3) the incentives to provide good service in a PaaS cloud where
  users have information about service reputation in selecting a
  provider.

\end{itemize}

The remainder of the paper is organized as follows. In
Section~\ref{sec:related} we survey related
work. Section~\ref{sec:motivation} highlights how a reputation-based
system can benefit different cloud environments, and also
identifies key design challenges. Section~\ref{sec:seit} describes the
architecture of Seit. Section~\ref{sec:impl} provides Seit's implementation details. 
We analyze the isolation and stability of
properties of Seit in Section~\ref{sec:properties}. We
evaluate Seit's efficacy in Section~\ref{sec:evaluation}. The paper
concludes in Section~\ref{sec:conclusion}.

%% file: related.tex
\section{Related Work}
\label{sec:related}

Seit builds on past research in cloud systems and reputation systems.
Here, we discuss these works.

\nip{Cloud Systems for Inter-tenant Communication.}  Communication
within clouds has largely focused on isolation mechanisms between
tenants.  A few systems focused explicitly on handling inter-tenant
communication.  In particular, CloudPolice~\cite{Popa2010cloudpolice}
implements a hypervisor-based access control mechanism, and
Hadrian~\cite{Ballani2013} proposed a new network sharing framework
which revisited the guarantees given to tenants.  Seit is largely
orthogonal to these.  Most closely related to Seit is
Jobber~\cite{jobber}, which proposed using reputation systems in the
cloud.  With Seit, we provide a practical implementation, demonstrate
stability and isolation, integrate with real cloud components, and
provide a full evaluation.

\nip{Reputation for Determining Communication.}  Leveraging reputation
has been explored in many areas to create robust, trust-worthy
systems. For example, Introduction-Based Routing (IBR) creates
incentives to separate misbehaving network participants by leveraging
implicit trust relationships and per-node
discretion~\cite{Frazier2011}.  With Seit, we leverage IBR in a
practical implementation for a cloud system and extend it to support more
dynamic use.  Ostra~\cite{mislove2008ostra} studied the use of trust
relationships among users, which already exist in many
applications (namely, social networks), to thwart unwanted communication. Ostra's credit scheme
ensures the overall credit balance unchanged at any times so that
malicious, colluding users can pass credits only between themselves,
protecting against sybils. SybilGuard~\cite{Yu2008SybilGuard} also
uses social networks to identify a user with multiple identities.  At
a high level, Ostra and Seit have some similarity, namely in the
objective of thwarting unwanted communication.  The biggest difference
is that Ostra was designed for person-to-person communication (email,
social networks, etc.), whereas Seit is designed for
machine-to-machine communication.  This then impacts how communication is handled.
In Ostra, communication is either wanted or unwanted; unwanted
communication is blocked. Seit goes beyond simply being able to
block certain users and allows for a variety of responses.  Further,
being machine-to-machine as opposed to being person-to-person impacts
how feedback is handled. In Ostra, feedback was explicit based on
human input, whereas Seit integrates with a variety of systems that
provide implicit feedback of varying strength.


\nip{Reputation for Peer-to-peer Systems.}  
In peer-to-peer systems,
reputation-based approaches involving other participants are used for
better decisions about cooperation such as preventing free-riders and
inauthentic file pieces in the system. For example,
EigenTrust~\cite{Kamvar2003eigentrust} aims to filter inauthentic file
pieces by maintaining a unique global reputation score of each peer
based on the peer's history of
uploads. Dandelion~\cite{Sirivianos2007dandelion},
PledgeRouter~\cite{Landa2009Sybil},
FairTorrent~\cite{Sherman2012fairtorrent}, and one hop reputation
system~\cite{Piatek2008Onehop} aim at minimizing the overhead of
maintaining reputation scores across peers either by placing a central
trusted server or by limiting the scope of reputation calculations.
Seit is targeting an environment where there is not a direct
give-and-take relationship; as such, we leverage a richer
reputation mechanism that maintains a graph of all user interactions
and uses it to determine local reputation for each user.

%% file: motivation.tex
\section{Reputation Matters}
\label{sec:motivation}

In this section, we elaborate on the potential benefits provided by a
cloud reputation-based system, covering four different
systems. Using these examples, we then identify the key design
challenges in building a cloud reputation-based system.

\subsection{Motivational Examples}
\label{sec:examples}

\nip{Reputation-augmented SDN Controller.} An IaaS cloud provides tenants with an
ability to dynamically launch and terminate virtual machines (VMs). The
provider's network is also becoming highly programmable using
SDN approaches.
%
%
SDN policies, in general, use {\em explicit} rules for
managing flows (blocking some, allowing some, or directing some
through a sequence of (security) middleboxes~\cite{casado2007ethane}).
%
A reputation-based system would extend the SDN interfaces to enable
flow control using {\em implicit} rules.  Instead of a tenant
needing to specify, for example, specific flows to block, it could
specify that flows originating from sources with low reputation scores
should be blocked.
This is illustrated in Figure~\ref{fig:ex_iaas}(a), where
Tenant $T2$ is blocking the traffic from $T1$.  Likewise, the tenant
can specify a set of middlebox traversal paths that then get tied to a
given reputation score. This is illustrated in
Figure~\ref{fig:ex_iaas}(b), where Tenant $T3$ views $T1$ as good, so
gives it a direct path; Tenant $T3$ also views $T2$ as suspect, so
forces its traffic through a middlebox.  

\begin{figure}
\centering
\includegraphics[width=\columnwidth]{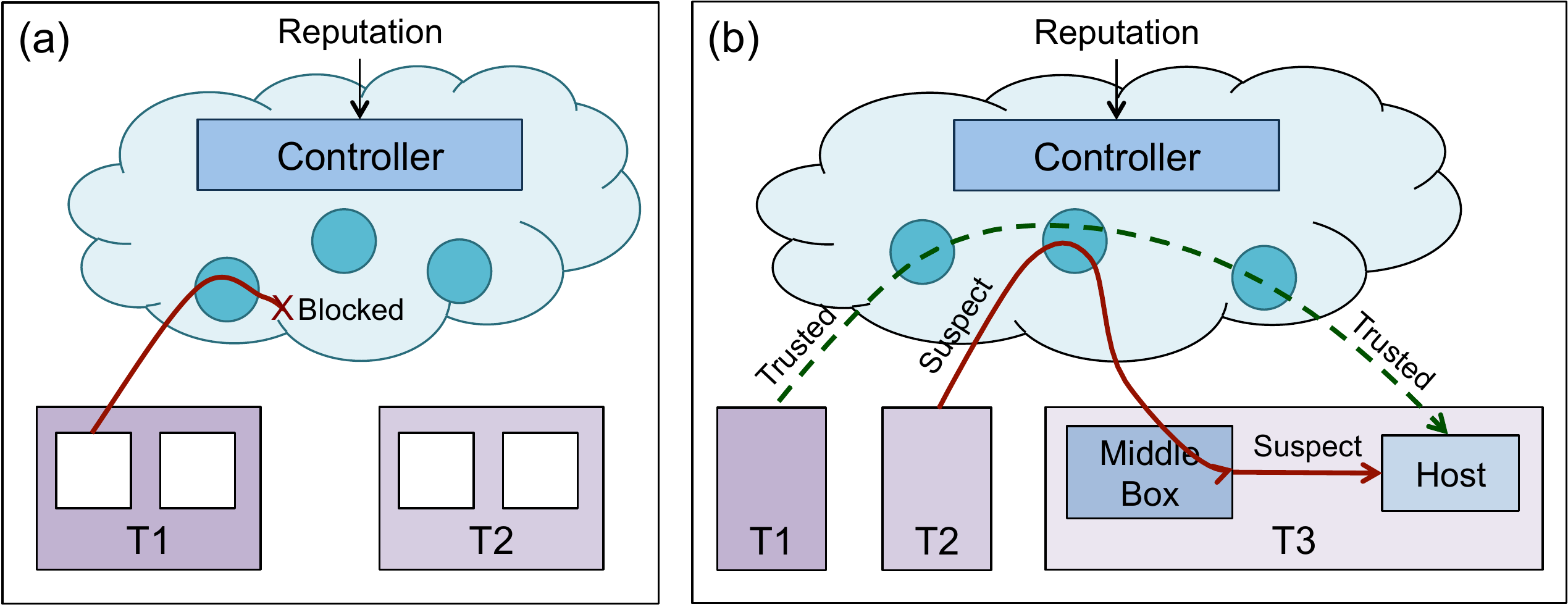}
\caption{Example of an IaaS provider using a reputation-based system}
\label{fig:ex_iaas}
\vspace{-0.1in}
\end{figure}

\nip{Reputation-augmented PaaS Brokers.} PaaS clouds offer the ability
to compose applications using services provided by the platform.  In
some cases, a PaaS cloud provides all of the services (\eg
Microsoft Azure~\cite{windowsazure}).  In other cases, platforms, such as
CloudFoundry~\cite{cloudfoundry}, provide an environment where many
service providers can offer their services.

As illustrated in Figure~\ref{fig:ex_paas}, service consumers use a
{\em broker} to (1) discover services and (2) bind to them. Service
discovery, in general, implements a simple search capability, focusing
on returning one-to-one match with the needed service (\eg version 2.6
of MongoDB).  With a reputation-based system, service discovery can be
enriched to include many new metrics that capture the quality of
service as perceived by other users. So, if two different providers
offer version 2.6 of MongoDB, then consumers can reason about which
service offer better quality.

In a similar way, a reputation-based system can be useful for service providers during the
binding phase, since it maintains historical information on consumers
of the service. In CloudFoundry, for example, the service provider
(\eg MongoDB) is responsible to implementing multi-tenancy. Most data
stores do not create separate containers (or VMs) per tenant; they
simply create different table spaces for each tenant. With a reputation-based system,
service providers can implement different tenant isolation
primitives. An untrusted tenant is provisioned a separate container
and is charged more because of the additional computing resources that
are required.

\begin{figure}
\centering
\includegraphics[width=0.9\columnwidth]{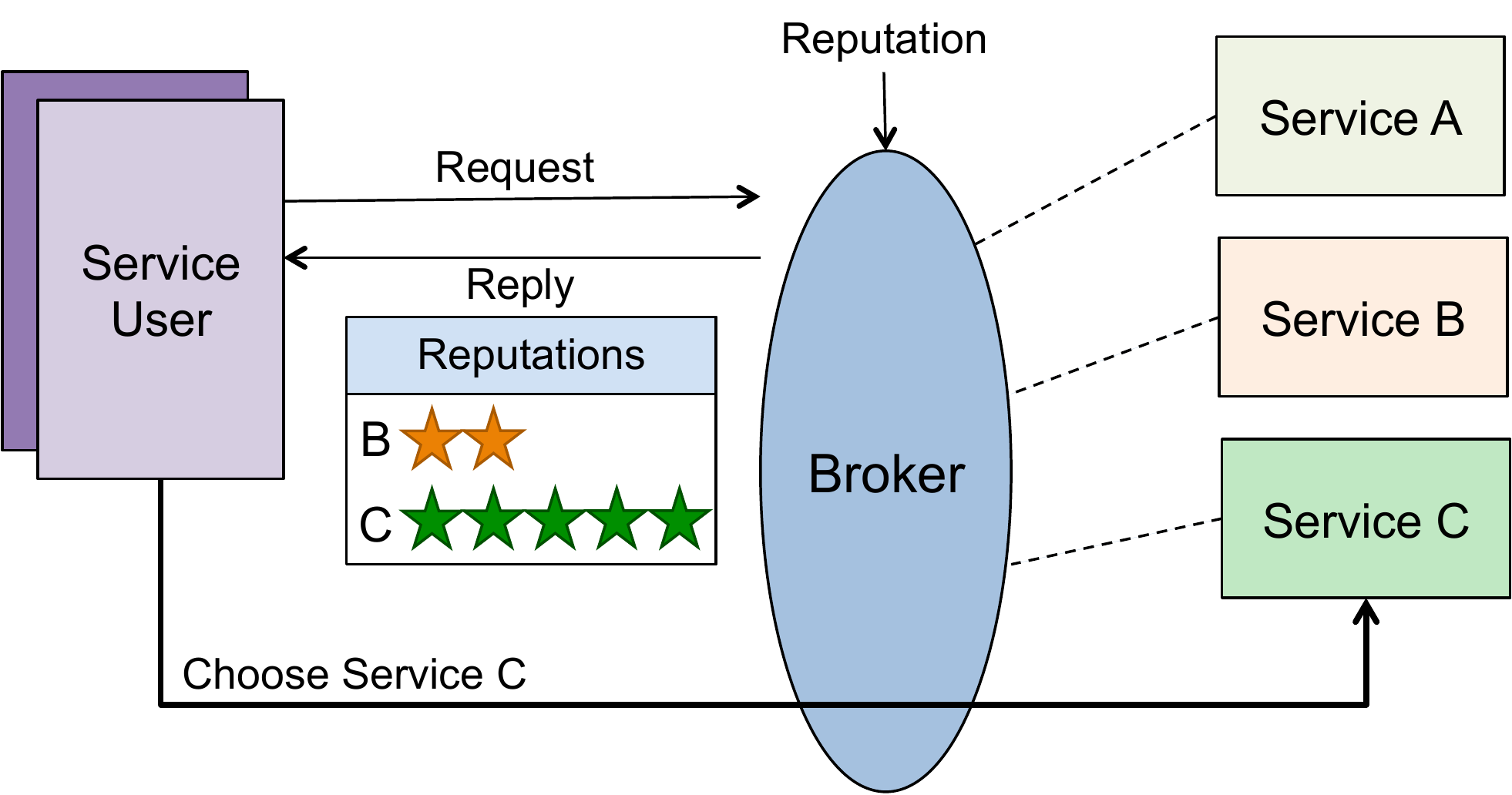}
\caption{Example of a PaaS provider using a reputation-based system}
\label{fig:ex_paas}
\vspace{-0.1in}
\end{figure}

\nip{Reputation-based Load Balancing.} Tenants and services can
directly implement a reputation-based system without explicit support from the cloud
providers. The resulting system in this setup would resemble a
peer-to-peer reputation-based system, where reputation is used to
provide differentiated services.

Figure~\ref{fig:ex_tenant} illustrates the integration of a reputation-based system into a
web service, where load balancers are used to distribute
client load across identical instances of the service.  Typically,
load balancers aim for even distribution~\cite{nginx}. 
With a reputation-based system, the web service can differentiate its users based on their
reputation, directing good/trusted clients to a set of
servers, and bad/untrusted clients to a different set of servers.


\nip{Sentiment-based Self-policing.} In traditional 
infrastructures, the administrator has great
visibility of what is happening inside of the infrastructure through a
variety of monitoring tools.  The administrator, however, has limited
visibility into how the infrastructure is viewed externally.  
Major outages are obvious and are easy to detect. Other
types of issues might, at best, result in an email (or other reporting
mechanism) being sent to the administrator.  
Sentiment analysis is widely used in corporations (\eg monitor Twitter feeds
to observe whether there is any positive or negative chatter affecting
its brand~\cite{henshen2012, o2010tweets}).  With a reputation-based system,
a service can monitor its sentiment as perceived by its consumers.  By
monitoring one's sentiment, the tenant can determine whether others
are having a negative experience interacting with it, then trigger a
root cause analysis.  This is supported by a report from
Verizon~\cite{verizon} which says that in many cases, infiltrations
are largely detected by external parties, not by the infected party
itself.

\subsection{Design Challenges}
\label{sec:challenges}

Reputation systems have been used in peer-to-peer systems to prevent
leachers~\cite{qiu2004modeling, Kamvar2003eigentrust,
Sirivianos2007dandelion, Sherman2012fairtorrent}, and in person-to-person communication
systems to prevent unwanted communication~\cite{mislove2008ostra}.  Applying to
the cloud has the unique challenges in being machine-to-machine,
highly variable, and highly dynamic interactions. In this subsection, we identify
five design challenges when building a cloud-based reputation-based
system.
 
\nip{Integration.} Cloud components and services come in all shapes and
sizes. Integrating a reputation-based system do not require
substantial development efforts. Similar to service life-cycle
calls in PaaS clouds~\cite{cloudfoundry}, a reputation-based system must define
simple, yet generic interfaces that can be easily implemented by
service providers. More importantly, the interfaces should support
configurable query language that provides support for efficient
continuous query and feedback.

\nip{Interpretation of a Reputation.}  Depending on the service,
reputations can be interpreted in many ways. Here, there is challenge
in defining what constitutes good or bad, and doing it automatically
(\ie a human will not be explicitly marking something as good or bad).
Even more, in machine-to-machine communication, an interaction is not
necessarily binary (good or bad), as there is a wide range of possible
interactions.

\begin{figure}[t]
\centering
\includegraphics[width=0.8\columnwidth]{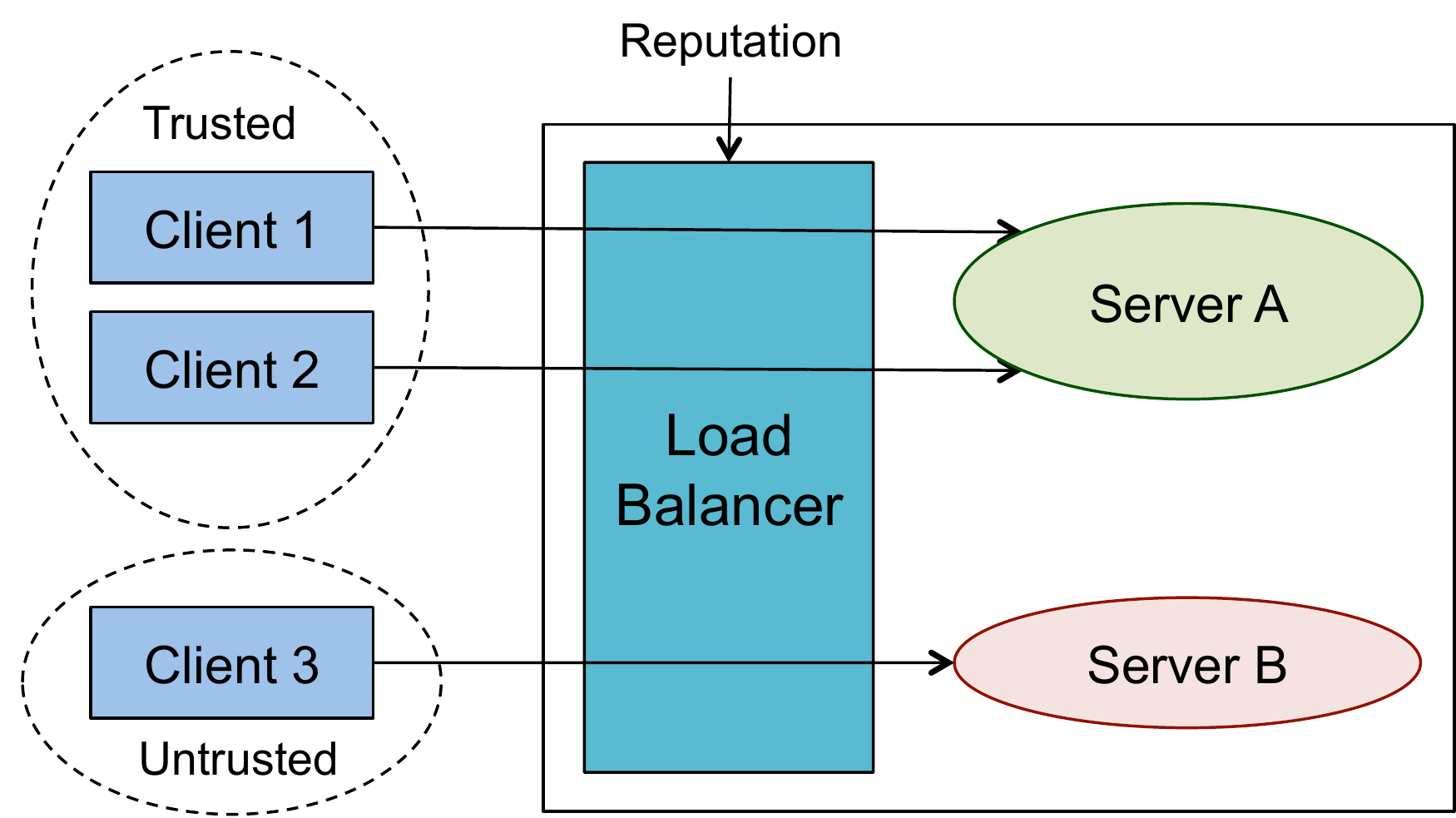}
\caption{Example of a tenant using a reputation-based system}
\label{fig:ex_tenant}
\vspace{-0.1in}
\end{figure}

\nip{Isolation.} In a human-centric reputation-based system (\eg Stack
Overflow~\cite{stack-overflow}), a global user reputation is
desirable. In contrast, the cloud consists of a wide variety of
systems. The reputation mechanism must be effective in clustering
tenants into groups (\eg to isolate bad tenants) based on both local
(tenant) view and global view.

\nip{Stability.} 
The ability to isolate bad tenants prevents system oscillations
as tenants adjust their reputations: instead, misbehaving tenants
will converge to a low reputation, and other tenants to a higher reputation.
Moreover, the reputation mechanism should be stable to short-term
fluctuations in behavior. For instance, if a tenant accidentally
misbehaves for a short time before resuming its normal behavior, it
should be able to eventually recover its reputation instead of being
immediately and permanently blacklisted by other tenants.

%

\nip{Resiliency.} 
Finally, a reputation mechanism must be resilient to attacks of the
reputation mechanism itself.  In particular, an attacker
falsely manages to build up a good reputation before launching an
attack by, for example, sybils, or other
tenants controlled by it that effectively say good things about the
attacker.

%% file: arch.tex
\section{Seit}
\label{sec:seit}

Seit was designed with the above challenges in mind. 
Figure~\ref{fig:arch} shows an overview of Seit's architecture.  Seit
includes a collection of interfaces that are specific to the
individual components and parties within the cloud.  Seit also consists of
a centralized\footnote{We capitalize on the centralized nature of the
  cloud for this implementation, but can envision a decentralized
  implementation which is less efficient, but equivalent in
  functionality.} reputation manager that interfaces with both the
cloud provider(s) and each tenant. In this section, we introduce these two
main components and describe how they address the above challenges.

\begin{figure}[t]
\centering
\includegraphics[width=0.75\columnwidth]{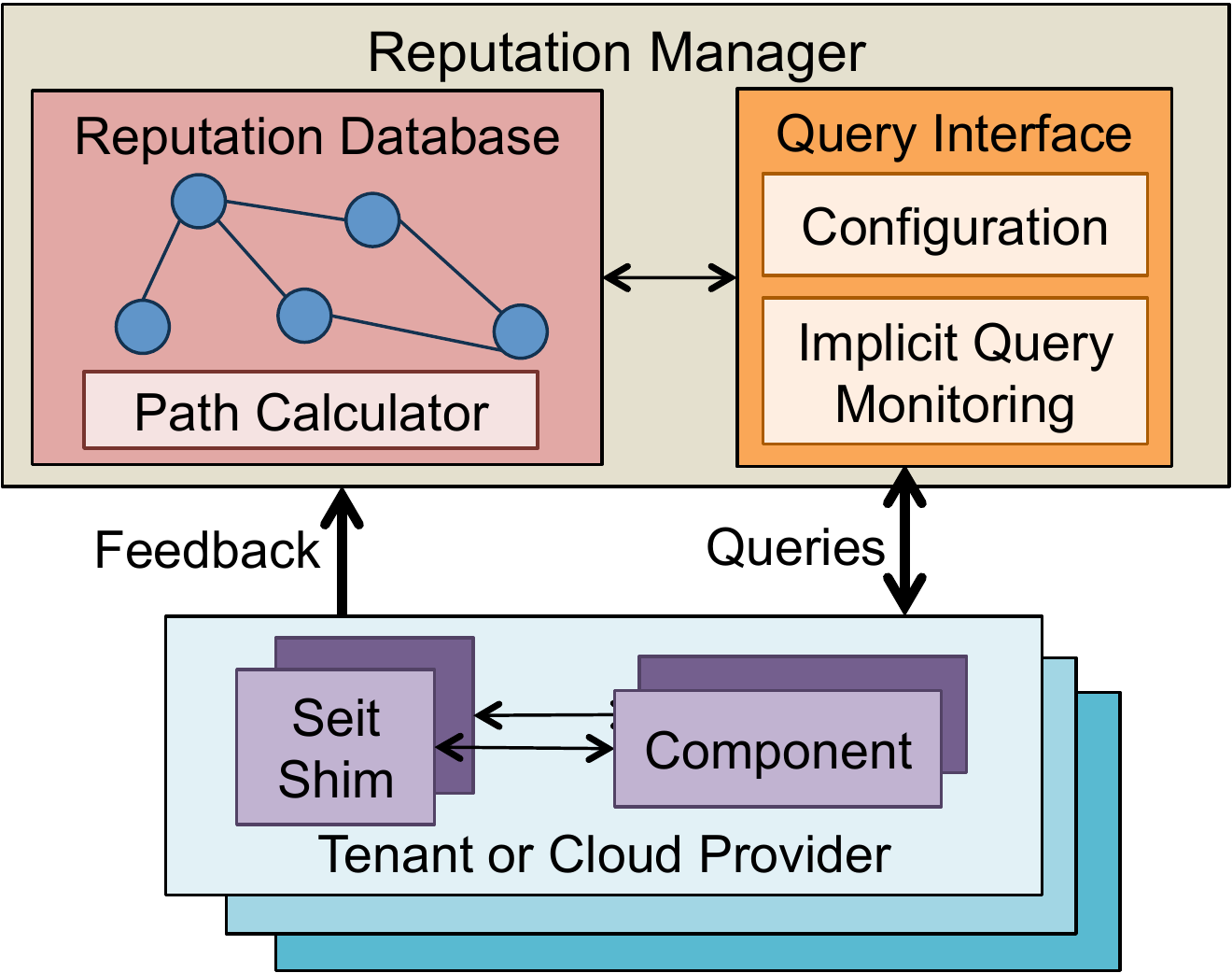}
\caption{Seit's architecture}
\label{fig:arch}
\end{figure}

%% file: policiesandinterfaces.tex
\subsection{Integration Interfaces}
\label{sec:interfaces}

To make Seit fully integratable into cloud systems, we need interfaces
between user components (\eg firewalls, load balancers, network
controllers, etc.) and the reputation manager.
Seit includes a framework to create a shim around existing components.
As shown in Figure~\ref{fig:shim}, 
%
Seit shims extend a component's existing interface with (potentially)
two additional interfaces that interact with Seit's reputation
manager.  The two additional interfaces represent two subcomponents:
inbound and outbound logic.  The shim's inbound logic
interprets how incoming reputation updates should impact the execution
of the component. The outbound logic translates how alerts,
events, and status updates from the component should impact the
reputation that is sent back to the reputation manager.
%
%
%
Both inbound and outbound logics are component specific; some
components only implement one of these two logics.  Here, we provide a few
examples to clarify the design and interface of shims:

\begin{itemize}

\item \textbf{SDN Controller in IaaS Clouds:} The IaaS network
  controller is what manages the physical cloud network.  We assume
  that the network controller has an interface to set up a logical
  topology (\eg place a firewall in between the external network and
  local network), and an interface to block traffic.\footnote{Even
    though former does not fully exist yet, we believe it will as the
    API becomes richer and as research in the space progresses.}
  The shim's inbound logic will extend these capabilities to
  make use of reputations.  For example, if reputation is less than 0,
  block; between 0 and 0.8, direct to a security middlebox; greater
  than 0.8, provide a direct connection.

\item \textbf{PaaS Broker:} The PaaS broker's responsibility is to
  effectively serve as a discovery mechanism for services in the
  cloud.  With Seit, we can extend it to enrich and filter the
  results.  Whenever a search request arrives at the broker, the
  shim's outbound logic would interpose on the request and queries
  the reputation manager to get the reputations of the searched services; the
  inbound logic would then sort and filter the results based on
  user-defined criteria. For example, it may be configured to filter
  out any services that would have less than a 0.3 reputation for a
  given user, and sort the remaining results.

\begin{figure}
\centering
\includegraphics[width=0.95\columnwidth]{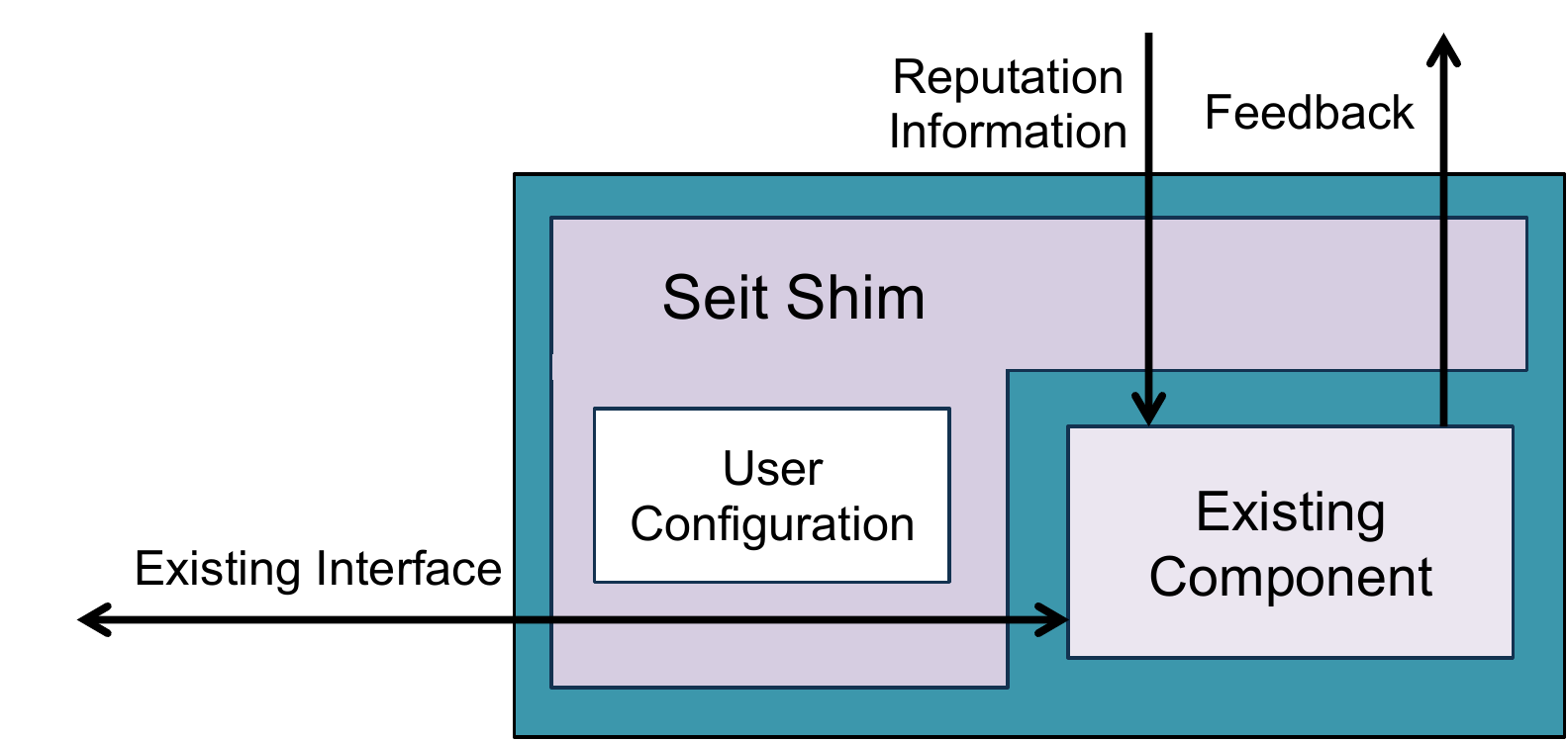}
\caption{Seit shim generic interface.}
\label{fig:shim}
\end{figure}


\item \textbf{Load Balancer:} As mentioned earlier, a
  reputation-augmented load balancer can be used to provide
  differentiated services to trusted and untrusted users. Here, the
  shim's inbound logic assigns new connections to servers based on the
  tenant's reputation score. 


\item \textbf{Infrastructure Monitor:} Infrastructure monitoring tools
  present information about the infrastructure to the
  administrators. Monitoring can be used as a way to alert
  administrators of changes in services reputations. This would be
  implemented in the shim's inbound logic. It can also be used to
  update the reputation of services based on monitoring information
  (\eg detecting port scans by a tenant). This would be implemented in
  the shim's outbound logic.

\item \textbf{Intrusion Detection System:} An intrusion detection
  system (IDS) monitors network traffic and looks for signatures
  within packets or performs behavioral analysis of the traffic to
  detect anomalies.  In this case, the shim's outbound logic is
  designed to intercept the alerts from the IDS, and allow users to
  configure the feedback weights for each alert type.  For example,
  the shim can decrease the tenant's reputation by 0.1 when seeing a
  connection drop alert.  Similarly, it can decrease the reputation by
  0.5 when seeing a port scan alert.

\end{itemize}

The above discussion presented only a few examples. We envision all
components being integrated with Seit to provide feedback.
For simplicity, the above discussion also focused on negative feedback.
Positive feedback is an important aspect as well.  Positive feedback
might be time-based, packet-based, or connection-based.
For example, a web server might provide positive feedback when
it goes through an entire session with well-formed http requests.
An IDS might provide positive feedback for every megabyte of traffic
that does not trigger an alert.

%% file: ibr.tex
\begin{figure*}[th]
\centering
\includegraphics[width=\textwidth]{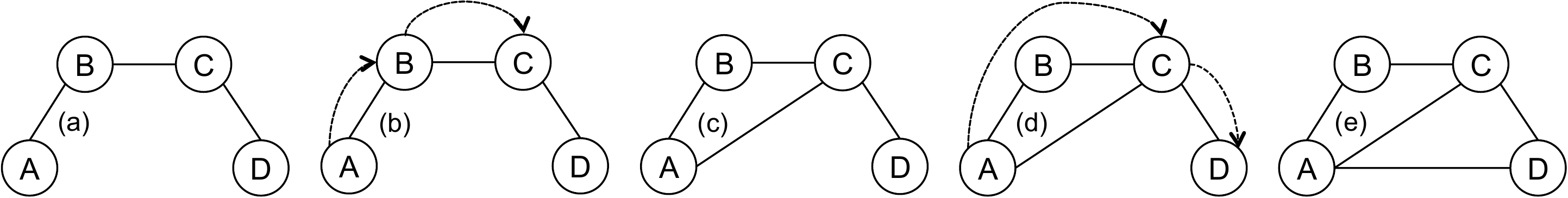}
\caption{Illustration of IBR: (a) Tenant $B$ is connected to Tenants $A$
  and $C$, and Tenant $D$ is connected to Tenant $C$, (b) $A$ asks $B$ to
  introduce it to $C$, and $B$ agrees to perform the introduction, (c) $C$
  accepts the introduction and starts communicating with $A$, (d) $A$ asks
  $C$ to introduce it to $D$, and $C$ agrees to perform this
  introduction, (e) $D$ accepts the introduction and starts
  communicating with $A$.}
\label{fig:IBR}
\end{figure*}

\subsection{Reputation Manager}
\label{sec:ibr}


The reputation manager is responsible for maintaining a view of the
reputations as perceived by various tenants.  The core of the
reputation manager is a reputation graph and means to query the
graph.

\subsubsection{Reputation Graph}

In Seit, reputation is modeled as a graph with nodes representing
tenants\footnote{In the future, we will explore reputations for a
  finer granularity than tenants.} and edges representing one node's
(tenant's) view of the other node (tenant) when the two have direct
communication with each other.  Here, we describe the mechanism for
building and using this reputation graph.

Seit adapts the introduction-based routing (IBR)
protocol~\cite{Frazier2011} used in P2P networks for tenant
interactions because of its ability to incorporate feedback in a
highly dynamic graph, and its resilience to sybil attacks.  IBR,
however, is not an intrinsic requirement of Seit. Seit requires only
that reputation scores are maintained for each tenant's view of other
tenants and uses these scores to determine the form of interaction
between them. We, thus, do not consider the full spectrum of
reputation system properties; considerations such as the possibility
of gaming the system are out of the scope of this paper.

 
\nip{Calculating Reputations.}
IBR in P2P networks allows peers to use participant feedback as a
basis for making their relationship choices.  While IBR can be
decentralized (as it was originally designed for P2P networks), in
centralizing the design, we internally use the IBR model, but
eliminate the signaling protocol between peers.  We give an example of
IBR's reputation calculation in Figure~\ref{fig:IBR}. The main idea is
that tenants can pass feedback to each other based on their
introductions to other tenants. Positive feedback results in a higher
reputation, and negative feedback in a lower reputation.

Nodes $A$, $B$, $C$ and $D$ in Figure~\ref{fig:IBR} are peers. The
straight lines indicate established connections. Each node maintains a
reputation score (trust level) for the nodes connected to it. When
Node $A$ wants to communicate with Node $D$, it must follow the chain
of connections between it and $D$ and ask the nodes in the chain for
introduction to the node after. Node $A$ starts asking $B$ for
introduction to $C$.  Node $B$ looks at node $A$'s behavior history
(represented by reputation score $AB$) from its local repository and
decides whether or not to forward $A$'s request to $C$. If the request
is forwarded, $C$ looks at behavior history of $B$ ($BC$) and decides
whether or not to accept the introduction request. The process
continues until $A$ reaches $D$. If $B$, $C$, or $D$ rejects a
request, node $A$ will not be able to communicate with $D$. After the
connection is established between $A$ and $C$, $C$ assigns a
reputations score to $A$ ($CA$) which is $x \times R(BC)$, where $x$
is a scaling parameter and $R(BC)$ is the reputation score of $B$ to
$A$.  Similarly, $D$ assigns reputations score of $DA$ which is $y
\times R(DC)$. If $A$ starts behaving negatively (\eg sending
malicious packets to $D$), $D$ will decrease $A$'s score $DA$ and also
decreases $C$'s score $DC$ since $C$ took the responsibility and
introduced $A$ to $D$. $C$ will do the same and decrease $A$'s score
$AC$ and $B$'s score since it introduced $A$ to $C$. Finally, $B$ will
decrease $A$'s reputation score $AB$. This approach ensures that nodes
are especially cautious about whom they introduce. Eventually,
misbehaving nodes will be isolated, with no other nodes will be
willing to introduce them when their scores fall under the minimum
trust level score. We will show this property in
Section~\ref{sec:properties}.

To reiterate, the IBR introduction mechanism is hidden from
tenants. They simply ask for a connection to another tenant and are
informed whether a path is available and accepted, and if so, what the
reputation score is.

\nip{Bootstrapping the Reputation Graph.}
When a tenant joins Seit's framework, it receives a default global
reputation score (assigned by the reputation manager) and zero local
reputation score until it begins to interacts with other tenants. Upon
being introduced to a tenant, the introduced tenant's initial
reputation score will be based on the introducer's score; it can then
evolve based on the reputation calculations that we describe next. 
In general, the speed with which a tenant builds a
reputation depends on several factors such as number of tenants it
interact with, the services it provides to these tenants, and any
privacy and security threats to these tenants. A new tenant does,
however, have an incentive to provide good services: its actions at
one tenant can propagate to others through introductions, influencing
its reputation at these other tenants. Since individual tenants do not
have a global view of these introduction relationships, these dynamics
also make it difficult for a malicious tenant to target a particular
victim, as it must first find an introduction chain to reach the
target.

%

\subsubsection{Configurable Query Interface}

In centralizing the IBR mechanism, we can provide a highly 
configurable interface to query the reputation.
Here, we elaborate on both the initial query configuration
as well as for subsequent queries.

\nip{Initial Query.} A reputation query can be as simple as
calculating the shortest path between two nodes, creating a new edge
between the nodes, and then only updating the direct edge between the
two nodes upon reputation feedback (direct feedback).  In the
reputation mechanism we are using, the feedback also impacts the edges
of the initial path between the two nodes (indirect feedback).

This adds an interesting aspect to the initial query: should an
intermediate node allow the search to include certain outgoing edges
in the query or not (or, in IBR terms, whether a node should make an
introduction).  To support each user's flexibility, Seit provides the
ability to configure two aspects:

\begin{itemize}

\item {\bf Outgoing Edge Selectivity:}  The idea behind introduction-based
routing is analogous to real life: if a person I trust introduces
someone, I am likely to trust that person more than if they were a
random stranger.  If I have a good interaction with that person, it
generally strengthens my trust in the person that made the
introduction.  A bad interaction can have the opposite effect.  As
such, people are generally selective in introductions, especially in
relationships where there is a great deal of good will built up.
In Seit, the tenant has full control over the thresholds for which an
introduction is made. They can introduce everyone with a low
threshold; they can also be selective with a very high threshold.

\item {\bf Query Rate Limit:} A consideration in serving as intermediate
nodes (making introductions) is the magnitude of the potential impact
to one's reputation.  For this, Seit includes the ability to limit the
rate at which a given node serves as an intermediate node.  In doing
so, it allows the system to adapt the reputations based on these new
interactions, so that future requests to serve as an intermediate node
will have that additional information.  As an extreme example, say
there is no rate limiting, Tenant $A$'s reputation is above the
threshold for Tenant $B$ to introduce it to Tenant $C$ through $Z$.  Tenant
$A$ then attacks $C$ through $Z$, and $B$'s reputation suffers accordingly.
Instead, if $B$ was rate limited, $A$ could only connect to $C$, and need to
build up good reputation with $C$, or wait sufficient time, to be able
to connect to the other tenants.  

\end{itemize}

\nip{Subsequent (implicit) Queries.} In Seit, we support the view that
any interaction should reflect the current reputation and that it is
the current reputation that should be used when handling any
interaction.  In other words, a reputation query should be performed
continuously for all ongoing interactions.  This, of course, would be
highly impractical.

Instead, Seit integrates triggers within the reputation manager.
Reputation changes whenever feedback is received---positive or
negative.  Within Seit, the effected paths through the graph that are
affected by any single edge update is tracked.  Then, upon an update of
an edge due to feedback being received, Seit will examine a list of
thresholds in order to notify a tenant when a threshold has been
crossed (which ultimately are sent to each component).

These threshold lists come from the shims within each tenant's
infrastructure.  The shims are what the tenant configures (or leaves
the defaults) to specify actions to take based upon different
reputation values.  When a shim is initialized, or the configuration
changes, the shim will notify the Seit reputation manager of these
values.


%% file: impl.tex
\section{Implementation}
\label{sec:impl}

We have built a prototype of the Seit reputation manager, and
integrated it with several cloud components.  We discuss these here.

\begin{table*}[t]
\small
\centering
\newcommand\T{\rule{0pt}{2.6ex}}
\newcommand\B{\rule[-1.2ex]{0pt}{0pt}}
\begin{tabular}{|m{1.2in}|m{1in}|m{4in}|}
\hline
\rowcolor[gray]{0.9} {\bf Category} & {\bf System} &  {\bf Description} \T \B \\
\hline
IaaS SDN Controller & Floodlight~\cite{floodlight} & The shim maps the reputation to OpenFlow rules via the Floodlight REST API to block or direct traffic. \\
\hline
PaaS Broker & CloudFoundry~\cite{cloudfoundry} & The shim interfaces between the CloudFoundry broker and the CloudFoundry  command line interface (used by the users) to filter and sort the marketplace results based on their reputation. \\
\hline
Load Balancer & HAProxy~\cite{haproxy} & This shim alters the configurations written in a haproxy.cfg file to specify load balancing based on the reputation (directing tenants to servers based on reputation).  Upon every change, the shim will tell HAProxy to reload the configuration. \\
\hline
Infrastructure Monitoring & Nagios~\cite{nagios} & We took advantage of JNRPE (Java Nagios Remote Plugin Executor)~\cite{jnrpe} to build a Java plugin that is listening for any reputations sent by Seit's reputation manager, and displays this sentiment and configures alerts for when sentiment (collective reputation of the tenant running Nagios) drops. \\
\hline
Intrusion Detection System & Snort~\cite{snort} & Snort alerts are configured to log to Syslog.  By using SWATCH~\cite{swatch} to monitor Syslog, the Seit shim is alerted to all Snort alerts.  The shim parses the alerts and extracts information such as source IP and alert type and send the feedback to the reputation manager. \\
\hline
\end{tabular}
\vspace{-0.1in}
\caption{Implemented shims}
\label{tab:plugins}
\end{table*}

We prototyped the reputation manager in approximately 7300 lines of
Java code.  We implemented the reputation manager as a scalable
Java-based server that uses Java NIO to efficiently handle a large
number of tenant connections. We also provided an admin API to setup,
install, and view policies across the cloud, as well as facilitate new
tenants and their services.

Rather than the reputation manager interfacing with each component,
within each tenant, we built a per-tenant server to serve as a proxy
between the tenant's components and the reputation manager. This proxy
 is a light weight Java process that can be installed on any
tenant machine. It listens on two separate interfaces for internal and
external communications. The internal interface is used to communicate
with a tenant's own components, while the
external interface is used to communicate with the reputation
manager. The proxy can be configured with a text configuration
file that specifies the following: (i) a list of components, each of
which has the name of the component, IP, type (service, executor,
sensor), component description and its tasks; (ii) the edge
selectivity threshold to specify when to refuse or accept connections
or introduction requests from a tenant; and (iii) the query rate
limit.


All communication in Seit is performed through a common messaging
interface. The API includes (1) a registration request from the
components when they boot up for the first time, (2) a connection
request when one tenant requires communication with another tenant,
which includes both the initial request and a message to approve (sent
to both source and destination) or reject the request (sent only to
the source), (3) a feedback message containing the components desire to
positively or negatively impact the reputation score for a given
connection (impacting both the reputation of the other tenant, but
also the introducers responsible), and (4) a configuration message
setting new thresholds and query configurations.

We built a shim interface for a number of cloud components, one for
each example discussed in Section~\ref{sec:interfaces}.
Table~\ref{tab:plugins} summarizes these components.

%% file: properties.tex
\section{Analysis of Seit's Isolation and Stability}
\label{sec:properties}

Using IBR allows Seit to both separate misbehaving tenants from well-behaved tenants and maintain stable reputations at each tenant. 
In this section, we formally show that these properties hold.


We consider an IBR system with $N$ tenants, each of whom desires services from other tenants and can provide some services in return.\footnote{Different tenants may provide different services; our analysis is agnostic to the type of the service.} We consider a series of discrete timeslots $t = 0,1,2,\ldots$ and use $q_{ij}[t] \in \left[-1,1\right]$ to denote tenant $i$'s feedback on the services provided by tenant $j$ to $i$ at time $t$. This feedback may include both the received fraction of tenant $i$'s requested service (e.g., 3GB out of a requested 5GB of SQL storage) as well as whether the service was useful (e.g., sending malicious packets). We let $R_{ij}[t] \in [0,1]$ denote tenant $j$'s reputation score at tenant $i$ during timeslot $t$. 


Tenants update their reputation scores in every timeslot according to feedback from the previous timeslot, as described in Section \ref{sec:ibr}. 
We suppose that these updates are linear in the feedback received, and define $q_{ij}^{\rm ibr}[t]$ as a weighted average of the feedback $q_{lk}[t]$ provided by all tenant pairs that contribute to $j$'s reputation at tenant $i$ (e.g., including tenants that $j$ introduced to $i$).
%
The reputation dynamics then follow
\begin{equation}
R_{ij}[t + 1] = \max\left\{(1 - \alpha)R_{ij}[t] + \alpha q_{ij}^{\rm ibr}[t], 0\right\}
\label{eq:ibr}
\end{equation}
where $\alpha\in (0,1)$ is a parameter chosen by tenant $i$. A larger $\alpha$ allows the reputations to evolve more quickly.

\nip{Isolation of Misbehaving Tenants.} 
In typical scenarios, tenants will likely act based on their reputation scores of other tenants: for instance, tenant $i$ would likely provide better service to tenants with higher reputation scores. We can approximate this behavior by supposing that the service provided by tenant $j$ to tenant $i$ (and thus $i$'s feedback $q_{ij}[t]$ on $j$'s service) is proportional to $i$'s reputation score at $j$: $q_{ij}[t] = \pm R_{ji}[t]$, where the sign of $q_{ij}[t]$ is fixed and determined by whether tenant $j$ is a ``good'' or ``bad'' tenant. The reputation dynamics (\ref{eq:ibr}) are then linear in the $R_{ij}$, allowing us to determine the equilibrium reputations:
\begin{prop}[Equilibrium Reputations]
\label{prop:equilibria}
Equilibria of (\ref{eq:ibr}) occur when, for each pair of tenants $(i,j)$, either $q_{ij}^{\rm ibr}[t] = R_{ij}[t]$ or $R_{ij}[t] = 0$ and $q_{ij}^{\rm ibr}[t] < 0$. These equilibria are Lyapunov-stable, and the system converges to this equilibrium. 
\end{prop}
\begin{proof}
We can find the equilibria specified by solving (\ref{eq:ibr}) with $q_{ij}[t] = \pm R_{ji}[t]$. To see that the equilibria are Lyapunov-stable, we note that (\ref{eq:ibr}) can be written as $R[t + 1] = \Sigma R[t]$, where $R[t]$ is a vector of the $R_{ij}[t]$ and $\Sigma$ a constant matrix. It therefore suffices to show that $\Sigma$ has no eigenvalue larger than 1. We now write $\Sigma = (1 - \alpha)I_{2N} + \alpha\Sigma_1$, where each row of $\Sigma_1$ sums to 1 since $q_{ij}^{\rm ibr}$ is a weighted average. 
Thus, the maximum eigenvalue of $\alpha\Sigma_1$ is $\alpha$, and that of $\Sigma$ is $(1 - \alpha) + \alpha = 1$. Since linear systems either diverge or converge to an equilibrium and the $R_{ij}$ are bounded, the system must converge to this (unique) equilibrium.
\end{proof}
This result shows that equilibria are reached when tenants agree with each others' reputations: the overall feedback $q_{ij}^{\rm ibr}[t]$ that tenant $i$ receives from tenant $j$ is consistent with tenant $j$'s reputation at tenant $i$.


We can interpret Prop. \ref{prop:equilibria}'s result as showing that at the equilibrium, tenants segregate into two different groups: one group of ``bad'' tenants who provide bad-quality service and have zero reputation, receiving no service; and one group of ``good'' tenants with positive reputations, who receive a positive amount of service.
Thus, tenants may experience a desirable ``race to the top:'' good tenants will receive good service from other good tenants, incentivizing them to provide even better service to these tenants. Bad tenants experience an analogous ``race to the bottom.''

We illustrate these findings in Figure \ref{fig:reputations}, which simulates the behavior of 100 tenants, 10 of which are assumed to be malicious $\left(q_{ij}[t] = -1\right)$. Tenants' reputations are assumed to be specified as in (\ref{eq:ibr}) and are randomly initialized between 0 and 1, with $\alpha = 0.1$. The figure shows the average reputation over time of both bad and good tenants. We see that good tenants consistently maintain high reputations at good tenants, while bad tenants quickly gain bad reputations at all tenants. 

\begin{figure}
\centering
\includegraphics[width = 0.45\textwidth]{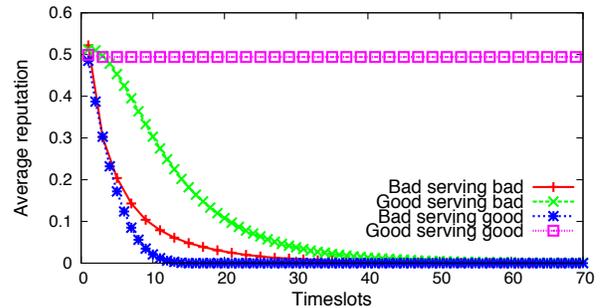}
\vspace{-0.05in}
\caption{Average reputations for ``good'' and ``bad'' tenants over time.}
\label{fig:reputations}
\vspace{-0.1in}
\end{figure}

\nip{Stability.}
While the analysis above considers binary ``bad'' and ``good'' tenants, some ``bad'' misbehaviors are not always malicious. For instance, tenants may occasionally send misconfigured packets by accident. Such tenants should be able to recover their reputations over time, instead of being immediately blacklisted at the client. Conversely, if malicious tenants occasionally send useful traffic in order to confuse their targets, they should not be able to improve their reputations permanently. We now show that this is the case:

\begin{prop}[Reputation Stability]
\label{prop:finite}
Let $\left\{R_{ij}[t],t \geq 0\right\}$ and $\left\{R'_{ij}[t],t \geq 0\right\}$ respectively denote the reputation scores given the feedback $q_{ij}^{\rm ibr}[t]$ and ${q_{ij}'}^{\rm ibr}[t]$. Suppose $q_{ij}^{\rm ibr}[0] \neq {q_{ij}'}^{\rm ibr}[0]$ and $q_{ij}[t] = q_{ij}'[t]$ for $t > 0$. Then $\lim_{t\rightarrow\infty} \left|R_{ij}[t] - R_{ij}'[t]\right| \rightarrow 0$. 
%
\end{prop}
The proof follows directly from (\ref{eq:ibr}). If tenants misbehave temporarily ($q_{ij}^{\rm ibr}[0] < 0$ and ${q_{ij}'}^{\rm ibr}[0] > 0$ but $q_{ij}^{\rm ibr}[t] > 0$), the effect of this initial misbehavior on their reputations disappears over time.


%% file: eval.tex
\section{Evaluation}
\label{sec:evaluation}

In this section, we evaluate both the performance of Seit's
implementation through micro-benchmarks as well as the benefits of
using reputation in a number of contexts. We used three large Linux
servers to run the experiments: {\em Server~1:} 64GB RAM, 24 Intel CPUs
(2.4GHz each), and running the reputation manager; {\em Server~2:} 64GB RAM,
24 CPUs (2.4GHz each), and running Mininet. {\em Server~3:} 32GB RAM, 12
Intel CPUs (2.00GHz each), and running Floodlight.

\subsection{Performance Overhead}

Despite its benefits, a reputation-based system does introduce some
overheads. In this subsection, we study the extent of these
overhead.

\nip{Query Throughput and Latency.} 
The reputation manager performs a query when a tenant wants 
to connect to another tenant.  We performed a micro benchmark 
where we varied the number of tenants in the network and calculated
both the throughput (number of queries per second) and latency (time
to calculate a single query on average) of our implementation.  Shown
in Figure~\ref{fig:throughput} is the throughput that the reputation manager
can handle for a given number of tenants.  To calculate throughput, we took a snapshot of 
the reputation graph from a simulated execution, and injected queries at a fixed rate.
The max throughput was the maximum query rate we were able to achieve such that the total
time to receive all responses was within a small threshold of the total time to send all queries 
(the response time is linearly related to the request time until overload, at which point
it becomes exponentially related).
Important to note is these results (i) reflects initial queries, 
which will not be frequent, and (ii) reflects a single instance. This can be 
mitigated if the reputation manager is designed as distributed component (and left as future work).
Shown in Figure~\ref{fig:latency} is the average latency of a single
query, on the order of milliseconds. This is similar (in order of magnitude)
to the overheads imposed, for example, by a typical SDN flow setup (and we
expect queries to be less frequent than flow setups).


\begin{figure}
\centering
\subcaptionbox{\label{fig:throughput}}{
  \includegraphics[width=0.45\textwidth]{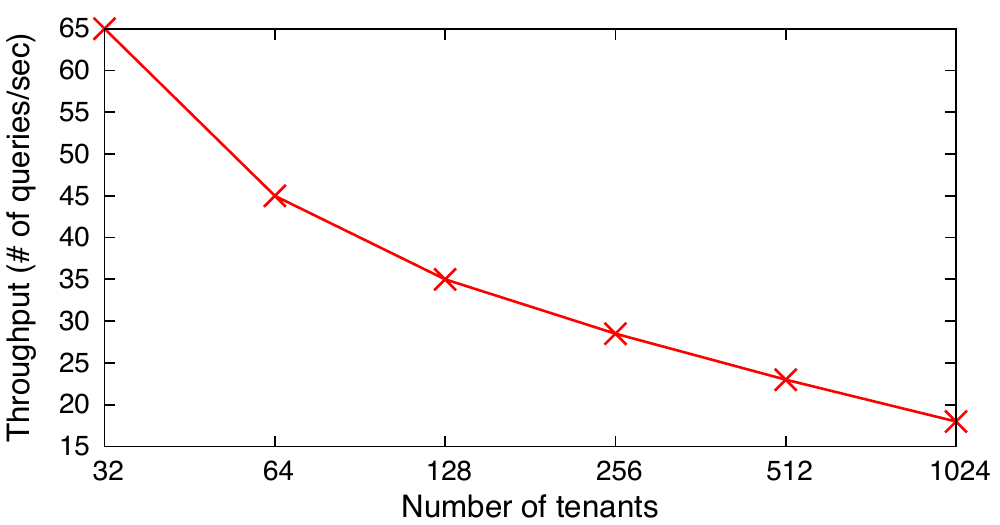}
  }\par\medskip   
\vspace{-0.1in}    
\subcaptionbox{\label{fig:latency}}{
  \includegraphics[width=0.45\textwidth]{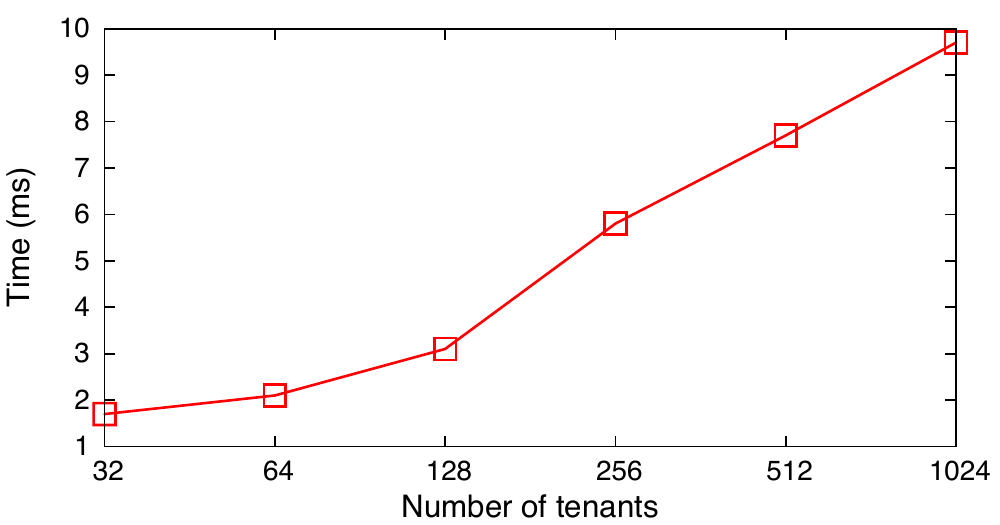}
   }
 \vspace{-0.1in}    
\caption{Reputation manager benchmark}
\label{TS}
\vspace{-0.2in}
\end{figure}

\nip{Impact of Dynamic Behavior.} To see how much dynamic behavior impacts an
example component, we varied the frequency of reputation change
notifications sent to an HAProxy component.  In our setup with HAProxy
running in its own virtual machine, iperf~\cite{iperf} reported
HAProxy with a static configuration as being able to achieve a rate of
8.08 Gbps.  With Seit updating the reputation (meaning the reputation
changed to cross a threshold) at a rate of once every second only
reduced the throughput of HAProxy to 7.78 Gbps; at (an extreme)
rate of once every millisecond, it reduced the throughput to 5.58 Gbps.

\subsection{Seit Benefits}


The main motivation for using Seit is that it can improve a variety of
aspects of a cloud operation. Here, we evaluate the benefit of Seit in
three contexts chosen to show: (i) security improvements, (ii)
efficiency gains (cost savings), and (iii) revenue gains.

\subsection{Set up and Parameters}

We built an evaluation platform using Mininet~\cite{mininet} to
emulate a typical cloud environment. In each experiment, we run the Seit reputation
manager along with configuring a tenant setup specific to each
experiment.  This evaluation platform allows us to specify four key
parts of an experiment:

\begin{itemize}
\item \emph{Graph Construction:} How the graph is built (\ie how
  interconnections are made).

\vspace{-1ex}
\item \emph{Sensor Configuration:} What is the sensor (\ie what does
  the sensor detect in order to provide feedback)

\vspace{-1ex}
\item \emph{Reputation Use:} What can be controlled (\ie what
  component/configuration does reputation impact).

\vspace{-1ex}
\item \emph{Traffic Pattern:} What is the traffic pattern.
  Importantly, in each case rates are chosen based on the emulation
  platform limitations (absolute values of rates are not meaningful).
\end{itemize}

\subsection{Improved Security by Isolating Attackers}
\label{sec:eval:iso}

One benefit of Seit is that it provides the ability to cluster good
participants and effectively isolate bad participants.  Here, we show
Seit's effectiveness in thwarting a simulated denial of service (DoS) attack,
where an attacker overwhelms his victims with packets/requests in
order to exhaust the victim's resources (or in the case of elastically
scalable services, cause the victim to spend more money).  In our
evaluation, we are mimicking an attack that happened on Amazon
EC2~\cite{ec2dos1, ec2dos2}, where hackers exploited a bug in
Amazon EC2 API to gain access to other tenants accounts and then flood
other servers with UDP packets.

We considered three scenarios in each run. The first scenario is a
data center where tenants do not use any kind of reputation feedback
and each tenant independently makes a local decision to block or allow
communication. In the other two scenarios, tenants use Seit, and thus
report about any attacks they detect and use reputation to isolate bad
tenants.  The only difference between the two scenarios is the attack
pattern. In one, the attacker will attack its victims sequentially;
in the other, the attacker establishes connections with all of its
victims simultaneously and attacks them in parallel.

\begin{itemize}

\item \emph{Graph Construction:}
For the graph construction, we select the number of tenants for
the given run of the experiment.  For each tenant, we select the number of
other tenants it connects to randomly from 1 to 5.  For each tenant, we set their `tenant quality': 3\% are 
explicitly marked as attackers with tenant quality of 0, the rest are randomly assigned
a quality metric from 0.1 to 1.0.

\vspace{-1ex}
\item \emph{Sensor Configuration:}
While Seit can support a variety of sensors,
in this experiment we use a simplified model,
where traffic explicitly consists of good packets and 
bad packets, and a simple sensor that detects ``bad'' 
packets with probability 0.9.

\vspace{-1ex}
\item \emph{Reputation Use:} the IaaS network controller
(in our case, a Floodlight controller) will block a tenant 
from being able to send
traffic to another tenant when the reputation of
the sender drops below some value.

\vspace{-1ex}
\item \emph{Traffic Pattern:} Tenants generate traffic according to a
simplistic model of a fixed rate and fixed inter-packet gap, where the
probability of sending a good packet or bad packet is based on the
tenant quality configuration (\eg $q<0.05$, representing a malicious
tenant, always send bad packets, $0.05<=q<0.5$, representing a
careless tenant, send bad packets infrequently, and in proportion to
the tenant quality, and $0.5<=q$, always send good packets).
Attackers send traffic at a rate 10 times higher than other tenants (1
every 10ms vs 1 every 100ms). Each attacker sends a total 100 packets
to each target while the rest send 50 packets.  Attackers deviate from
the connection graph determined above, instead attempt connection
to 25\% of the other tenants randomly.

\end{itemize}

We measured the total number of attack packets generated by attackers
and the total number of these packets that reached the victims.  We
varied the total number of tenants in each run starting from 32
tenants up to 1024 tenants.

As shown in Figure \ref{fig:ddos}, without Seit, over 90\% of the
attack packets will reach the victims, overloading the security
middleboxes. With Seit, on the other hand, we are able to isolate the
attacker.  With more tenants in the cloud, we should be able to block more
attack traffic as there will be greater amount of information about
the attacker.  A byproduct that is not shown is
that, in this experiment, we are also decreasing the total overall
traffic on the network by blocking at the source.

\begin{figure}
\centering
\includegraphics[width =\columnwidth]{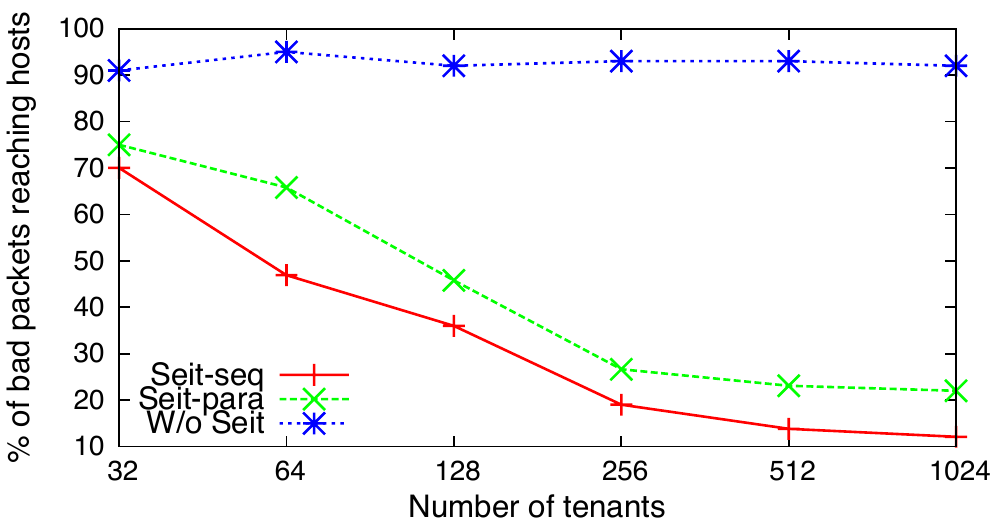}
\vspace{-0.1in}
\caption{Effectiveness of protecting tenants in an IaaS cloud}
\label{fig:ddos}
\vspace{-0.05in}
\end{figure}

\subsection{Decreased Costs by Managing Middlebox Chaining Policy}
Another benefit of being able to differentiate users with reputation
is that we can decrease the cost of operating security middleboxes,
without compromising security.  Here, we explore this benefit:

\begin{itemize}

\item \emph{Graph construction:}
The graph construction is identical to the experiment in Section~\ref{sec:eval:iso}, 
except that we fix the total number of tenants at 1024.

\vspace{-1ex}
\item \emph{Sensor configuration:}
The sensor configuration is identical to the experiment in Section~\ref{sec:eval:iso}.

\vspace{-1ex}
\item \emph{Reputation Use:} Each tenant will direct other tenants
connecting to it either through a middlebox or allow the tenant to
bypass the middlebox based on the reputation of the connecting
tenant. This represents a simple middlebox chaining policy.

\vspace{-1ex}
\item \emph{Traffic Pattern:} The traffic pattern is identical to the
experiment in Section~\ref{sec:eval:iso}.

\end{itemize}

We place the constraint that a single middlebox instance can only handle 10 packets per second.
This then allows us to capture the tradeoff between cost and security effectiveness
(in our experiments, measured as number of bad packets that ultimately reached an end host). 

We ran two variants of this experiment.  In one variant we allow the
number of middleboxes to scale to what is needed to match the traffic
in a given time interval.  In the other variant, we fix the budget to
a specific number of middleboxes, in which case, if the middleboxes
are overloaded, they will fail to process every packet.  In each case,
we calculate the total cost of operation (number of middlebox
instances needed) as well as the security effectiveness (percentage of
attack packets reached destination host).  As shown in
Figure~\ref{fig:resource_vs_security}, we can see that using Seit has
a distinct improvement in security when being held to a fixed budget,
and a distinct reduction in cost when shooting for a specific security
coverage to handle the varying load.

%

\begin{figure}
\centering
\subcaptionbox{\label{fig:middlebox_cost}}{
  \includegraphics[width=0.45\textwidth]{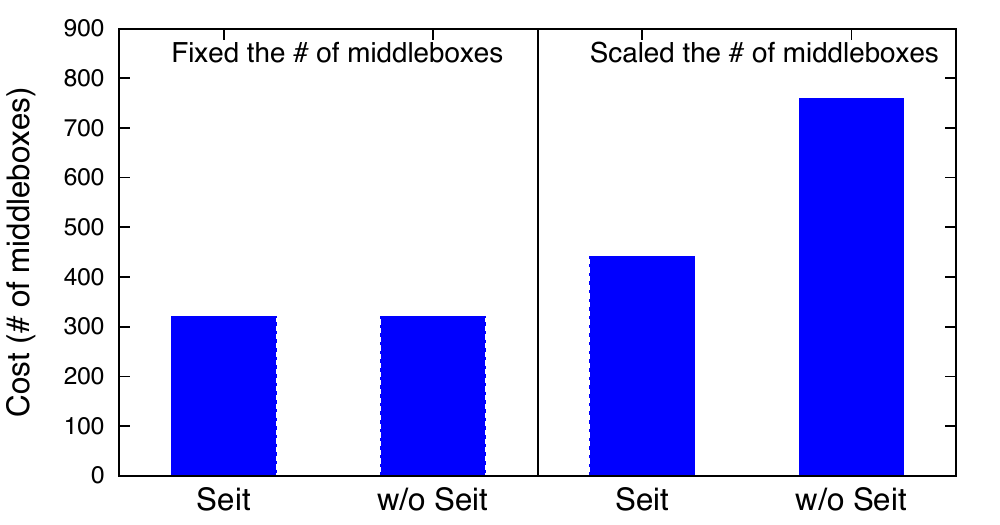}
  }\par\medskip 
  \vspace{-0.1in}      
\subcaptionbox{\label{fig:attack_missed}}{
  \includegraphics[width=0.45\textwidth]{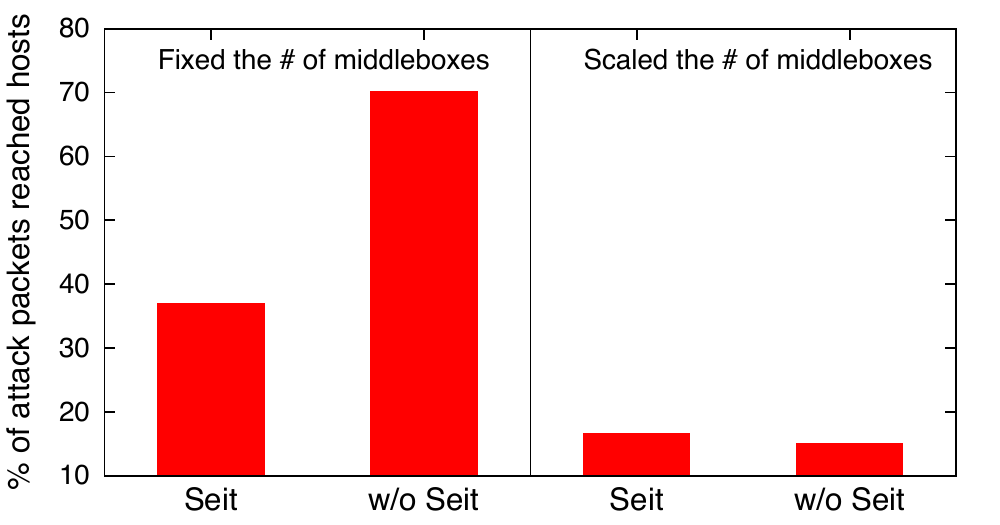}
  }
\vspace{-0.1in}
\caption{Resource saving and security}
\label{fig:resource_vs_security}
\vspace{-0.1in}
\end{figure}

\subsection{Increasing Revenue by Managing PaaS Broker Search}

In PaaS clouds, such as CloudFoundry, service providers offering
similar service need a way to differentiate themselves.  With a
reputation system as provided with Seit, service providers 
that have the highest quality of service get rewarded with more
customers.  Using the Seit integration with the CloudFoundry broker
which sorts and filters search results, we evaluate the relationship
between quality of service and revenue.

\begin{itemize}

\item \emph{Graph Construction:} In this experiment, we have 1024 tenants,
where we selected 256 tenants as service providers, 256 tenants
as both service provider and service users, and the rest as service
users only. For simplicity, we assume all service providers are
providing identical services.  To distinguish between service
providers, we use four discrete tenant quality values 0.2, 0.4, 0.6,
and 0.8 (higher is better).  To bootstrap the experiment, we create an
initial graph where for each service user tenant, we randomly select the number
of service provider tenants it connects to (from 1 to 5).

\vspace{-1ex} 
\item \emph{Sensor Configuration:} Here, clients make requests and receive
responses.  The sensor detects whether a request got a response or not;
Dropped requests are proxy for poor service.

\vspace{-1ex} 
\item \emph{Reputation Use:} For a PaaS broker, service tenant users perform
a search for services to use. In the broker, we filter and sort the
search results according to the service provider's reputation.  We
assume a client performs a new search every 20 seconds.  The service
user will choose among the top results with some probability
distribution (1st search result chosen 85\% of the time, 2nd result
10\% of the time, 3rd result 5\% of the time).

\vspace{-1ex} 
\item \emph{Traffic Pattern:} As we are attempting to show the incentive to a
service provider, we make a simplifying assumption that a service user
always sends a good request.  The service provider will either respond
or not based on the tenant quality for the given service provider.
Every tenant sends a packet every second to the service provider
connecting to. If a tenant receives a response, it will increase the
reputation of the service provider and update its neighbor with this
information. If no response received, the same process will happen but
with a negative impact.

\end{itemize}

We run the experiment for two minutes, and as a proxy for revenue, we
count the number of times a service user selects a given service
provider. As shown in Figure~\ref{fig:revenue}, the expected benefits hold. As
tenants with the greatest tenant quality (0.8) had a greater revenue
(showing over 2000 times they were selected), while tenants with
lowest tenant quality had the least revenue (being selected around 300
times).

\begin{figure}
\centering
\includegraphics[width =\columnwidth]{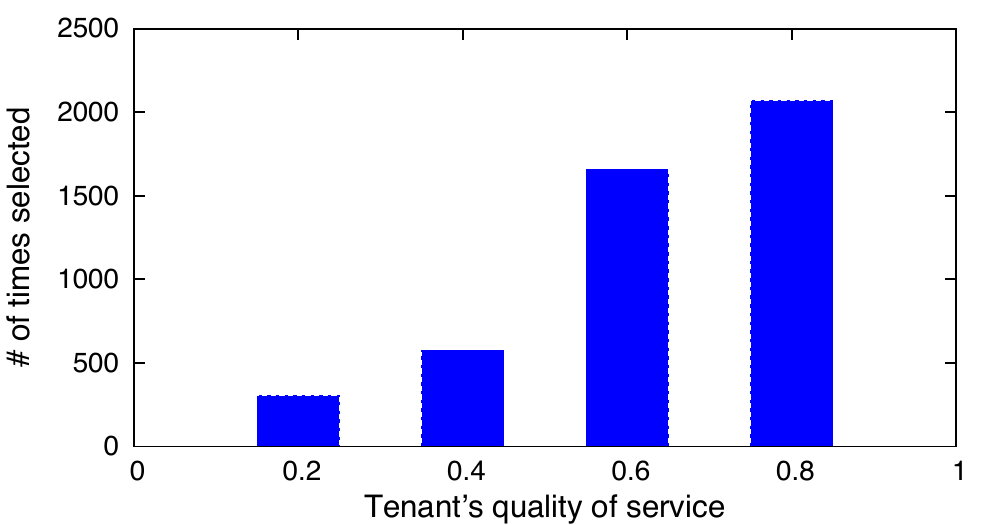}
\caption{Revenue increase for high quality service providers}
\label{fig:revenue}
\vspace{-0.1in}
\end{figure}

%% file: concl.tex
\section{Conclusions and Future Work}
\label{sec:conclusion}

Cloud systems today are fostering ecosystems of interacting services.
In this paper, we presented {\em Seit} as an inter-tenant framework
that manages the interactions within the cloud through the use of a
reputation-based system.  The Seit architecture overcomes key
challenges of using a reputation system in cloud environments around
integration, isolation, stability, and resiliency. Using practical
implementation, we demonstrate Seit's benefits across a wide spectrum
of cloud services.  As future work, we plan to integrate more
components and improve the overall performance of Seit. We also want
study the implications of incremental deployments, where some tenants do
not implement Seit.  Finally, we want to study scalability challenges when
managing a large number of components and tenants, especially across
autonomous geo-distributed clouds.

\vspace{0.1in}
